\documentclass{birkjour}

\usepackage{amsmath}
\usepackage{hyperref,upref} 
\hypersetup{citebordercolor=.1 .6 1}

\IfFileExists{mmymtpro2.sty}{%
  \usepackage[mtpcal,mtpscr,subscriptcorrection]{mymtpro2}
	\newcommand{\one}{\mathbb{1}}
	\renewcommand{\cup}{\cupprod}
	\renewcommand{\cap}{\capprod}
	\renewcommand{\enlargethispage}[1]{}
}%
{%
  \usepackage{amssymb}
  \usepackage{mathrsfs}
}

\newtheorem{theorem}{Theorem}[section]

\newtheorem{lemma}[theorem]{Lemma}
\theoremstyle{definition}
\newtheorem{definition}[theorem]{Definition}
\newtheorem{remark}[theorem]{Remark}

\numberwithin{equation}{section} % to reset the numbering counter at the beginning of each section and the inclusion of the %section number in the equation numbering

\newcommand{\HAM}{hierarchical Anderson model}
\newcommand{\LT}{Lifshits tails}
\newcommand{\RV} {random variables}
\newcommand{\hl}{hierarchical Laplacian}
\newcommand{\sa}{self-adjoint}

\let\punkt\cdot %in order to redefine \cdot 1
\renewcommand{\cdot}{\boldsymbol{\punkt}} %in order to redefine \cdot 2

\newcommand{\PP}{\mathbb{P}}
\newcommand{\X}{\mathbb{X}} 
\newcommand{\E}{\mathbf{E}}
\newcommand{\Q}{Q_\kappa}
\newcommand{\N}{{N}}
\newcommand{\Nk}{\mathrm{N}, \kappa}
\newcommand{\Dk}{\mathrm{D}, \kappa}
\newcommand{\Xk}{\mathrm{X}, \kappa}
\newcommand{\Nr}{\mathrm{N}, r}
\newcommand{\Dr}{\mathrm{D}, r}
\newcommand{\Xr}{\mathrm{X}, r}
\newcommand{\Ha}{{H}} %\newcommand{\Ha}{\mathcal{H}}
\newcommand{\Han}{\Ha^{\omega}_{\Nk}}
\newcommand{\Had}{\Ha^{\omega}_{\Dk}}

\newcommand{\Hd}{\wtilde{\Ha}_{\Dk}^{\omega}}
\newcommand{\Emd}{E_{\mathrm{max}}(\Had)}
\newcommand{\Emn}{E_{\mathrm{max}}(\Han)}
\newcommand{\nn}{\nonumber}
\newcommand{\diagdots}{\reflectbox{$\ddots$}} %für umgekehrte \ddots in Figure 1
\newcommand{\Chi}{\raisebox{.2ex}{$\chi$}}

\renewcommand{\le}{\leqslant} \let\leq\le
\renewcommand{\ge}{\geqslant} \let\geq\ge
\renewcommand{\emptyset}{\varnothing}
\providecommand{\one}{\mathbf{1}}
\providecommand{\wtilde}[1]{\widetilde{#1}}
\providecommand{\bigtimes}{\mathop{\raisebox{-1pt}{\LARGE $\times$}}}

\DeclareMathOperator{\e}{e} 
\DeclareMathOperator{\supp}{supp}
\DeclareMathOperator{\spec}{spec}

\DeclareMathOperator{\cohu}{ch}
\DeclareMathOperator{\tr}{tr}
\DeclareMathOperator{\sydi}{\raisebox{.6pt}{\scriptsize$\bigtriangleup$}}

\def\pper{.}
\def\HarvardComma{}

\newcounter{aucount}
\newif\ifedplural
\newif\ifper\pertrue

\def\au#1#2{{#1 #2}}
\def\lau#1#2{{#1 #2}, }
\def\ed#1#2{\ifnum\theaucount=0\relax\fi{#1 #2}\addtocounter{aucount}{1}}
\def\led#1#2{\ifnum\theaucount=0\relax\edpluralfalse\else\edpluraltrue\fi{#1
    #2} (\editorname.),\setcounter{aucount}{0}}
\def\editorname{\ifedplural Eds\else Ed\fi}

\def\et{\ifnum\theaucount=1\else\HarvardComma\fi{} and\ }
\def\ti#1{\emph{#1}.\ifper\fi\pertrue}

\makeatletter

\def\bti{\@ifnextchar[\bbti\bbbti}
\def\bbti[#1]#2{\emph{#2}, #1.}
\def\bbbti#1{\emph{#1}.}

\def\z{\@ifnextchar[\zz\zzz}
\def\zz[#1]#2#3#4#5{\perfalse\emph{#2}\ifx @#3@\relax\else{} \textbf{#3}\fi, #4 \ifx
  @#5@\relax\else (#5)\fi [#1]\ifper\pper\fi\pertrue} 
\def\zzz#1#2#3#4{{#1}\ifx @#2@\relax\else{} \textbf{#2}\fi, #3 \ifx @#4@\relax\else
  (#4)\fi\ifper\pper\fi\pertrue}

\def\pub{\@ifstar\pubstar\pubnostar}
\def\pubnostar{\@ifnextchar[\@@pubnostar\@pubnostar}
\def\@@pubnostar[#1]#2#3#4{#2, #3, #4, #1\ifper\pper\fi\pertrue}
\def\@pubnostar#1#2#3{#1, #2, #3\ifper\pper\fi\pertrue}
\def\pubstar[#1]#2#3#4{\perfalse #2, #3, #4 [#1]\pper\pertrue}

\def\IN{\@ifnextchar[\IIN\IIIN}
\def\IIN[#1]{Pp. #1 in }
\def\IIIN{In }

\makeatother
\makeatletter
\def\ref#1{{\upshape{\@newref{#1}}}}
\def\@cite#1#2{{\m@th\upshape[{#1\if@tempswa, #2\fi}]}}
\makeatother
\begin{document}

\title{\raggedright \LT{} in the \HAM}

\author{Simon Kuttruf}
\address{Mathematisches Institut, Universit\"at M\"unchen,
	Theresienstra\ss{}e 39, 80333 M\"unchen, Germany}

\author{Peter M\"uller}
\address{Mathematisches Institut, Universit\"at M\"unchen,
	Theresienstra\ss{}e 39, 80333 M\"unchen, Germany}
\email{mueller@lmu.de} 

\thanks{\emph{Version of 8 August 2011} -- to appear in Ann.\ Henri Poincar\'e (2011). The final publication is available at www.springerlink.com, \href{http://dx.doi.org/10.1007/s00023-011-0132-1}{\texttt{DOI 10.1007/s00023-011-0132-1}}\,.\\ 
	This work has been partially supported by Sfb/Tr 12 of the German Research Foundation (DFG). Simon Kuttruf dedicates this work to Mark.}



\begin{abstract}
	We prove that the homogeneous hierarchical Anderson model exhibits a Lifshits tail at the upper 
	edge of its spectrum. The Lifshits exponent is given in terms of the spectral dimension 
	of the homogeneous
	hierarchical structure. Our approach is based on Dirichlet-Neumann bracketing for the 
	hierarchical Laplacian and a large-deviation argument.
\end{abstract}

\maketitle

\section{Introduction}
\label{intro}

Hierarchical models have a long tradition in statistical physics. Dyson \cite{Dys69, Dys71} introduced them as an auxiliary tool in his study of phase transitions in the one-dimensional Ising ferromagnet with long-range interactions. 
An important feature of hierarchical models is that they preserve their structure under renormalisation-group transformations.
Bleher and Sinai \cite{BS, BlSi75} exploited this to determine critical properties of hierarchical spin models.  

Bovier \cite{Bovier} seems to be the first who studied the hierarchical Anderson model, that is, the Anderson model on a countably infinite configuration space with kinetic energy given by the hierarchical Laplacian. He also pursued a renormalisation-group approach and showed analyticity properties of the density of states. Under very mild conditions, Molchanov \cite{Mol94, M} established that the \HAM{} with Cauchy-distributed random variables has only pure point spectrum. In particular, his result does not require a homogeneous hierarchical structure. More recently, Kritchevski \cite{K1, K2, K3} continued the investigations of the hierarchical Anderson model. In \cite{K1, K2} he removed the requirement for a Cauchy distribution and proved Anderson localisation at all energies and for general single-site distributions. However, his proof only works for homogeneous hierarchical structures with spectral dimension $d_{s} \le 4$ (the spectral dimension will be introduced in Definition~\ref{d-s-def} below).
In \cite{K3} he proved Poisson statistics of rescaled eigenvalue distributions for the homogeneous hierarchical Anderson model with $d_{s} <1$. If the hierarchical Laplacian is modified as to contain also suitable negative hopping rates, then Monthus and Garel \cite{MoGa11} argue in favour of a localisation-delocalisation transition in the hierarchical Anderson model. Fyodorov, Ossipov and Rodriguez found analytical and numerical evidence for a localisation-delocalisation transition for random matrices with a hierarchical structure of the correlations between the matrix entries \cite{FyOsRo09}. 

In this paper we prove the occurrence of a Lifshits tail at the upper spectral edge of the hierarchical Anderson model. More precisely, we only deal with homogeneous hierarchical structures for which one can define a spectral dimension $d_{s}$ (but we do not impose any restriction on the value of $d_{s}$).  In particular, we find that the Lifshits exponent of the integrated density of states of the hierarchical Anderson model coincides with $d_{s}/2$. Technically, we follow the approach that was successfully employed for Poisson and alloy-type random Schr\"odinger operators by Kirsch and Martinelli \cite{KM}, and for the Anderson model on the lattice by Simon \cite{Simon}, see also the recent survey by Kirsch \cite{Kirsch}. The method requires Dirichlet-Neumann bracketing for finite-volume operators which we establish for suitable finite-volume restrictions of the hierarchical Laplacian. 
 
The quantity $d_{s}/2$, which we find for the Lifshits exponent in Theorem~\ref{lif-thm}, is also referred to as the van Hove exponent, since it governs the van Hove ``singularities'' of the integrated density of states in the absence of disorder, 
see also Lemma~\ref{vanHove}. 
Amazingly, the equality of the Lifshits and van Hove exponent is known to hold for very different types of random hopping models. First and foremost we mention the standard alloy-type random Schr\"odinger operators in $\mathbb{R}^{d}$ or $\mathbb{Z}^{d}$, $d\in\mathbb{N}$. In that case $d_{s}=d$, and the Lifshits exponent equals $d/2$, see e.g.\ \cite{KM, Simon} and references therein. But it also holds for the integrated density of states of the Dirichlet Laplacian on percolation subgraphs. This was first shown for bond percolation on $\mathbb{Z}^{d}$ in \cite{KiMu06, MuSt07} and then generalised to a large class of Cayley graphs in \cite{AV}, see also the recent review \cite{MuSt11}. From this perspective, Theorem~\ref{lif-thm} is interesting because it establishes the equality of the Lifshits and van Hove exponent for a random perturbation of the rather peculiar hierarchical Laplacian. The deeper reason behind it in all mentioned models is the intuition explained in Remark~\ref{strategy}.

The paper is organised as follows. After introducing our notation and presenting the main result, Theorem~\ref{lif-thm}, in Sect.~\ref{HAM}, we turn to Dirichlet-Neumann bracketing for hierarchical finite-volume operators in Sect.~\ref{Sandwich}. Finally, Sect.~\ref{LT} is devoted to the proof of Theorem~\ref{lif-thm}, and the appendix compiles the ergodic structure of the \HAM.

\enlargethispage{1ex}

%%%%%%%%%%%%%%%%%%%%%%%%%%%%%%%%%%%%%%%%%%%%%%%%%%%%%%%%%%%%%%%%%%%%%%%%%%%%%
%%%%%%%%%%%%%%%%%%%%%%%%%%%%%%%%%%%%%%%%%%%%%%%%%%%%%%%%%%%%%%%%%%%%%%%%%%%%%

\section{Model and result}\label{HAM}

We consider a countably infinite configuration space $\X$ and the quantum Hamiltonian
\begin{equation}\label{Hamiltonian}
\Ha^\omega := \Delta + V^\omega
\end{equation}
acting on the Hilbert space $\ell^2(\X)$ of complex-valued, square-summable sequences over $\X$. The Hamiltonian describes diagonal disorder through its potential energy, which acts as the multiplication operator
\begin{equation}
(V^{\omega}\psi)(x):=\omega_x \,\psi(x)\quad \text{~for all~} \psi\in \ell^2(\X)\, \text{~and all~} x\in\X.
\end{equation}
Here, $\omega:=(\omega_x)_{x\in\X}$ is a family of independent and identically distributed (i.i.d.), real-valued random variables. We think of them as being canonically realised in the probability space $\Omega:=\mathbb{R}^{\X}$, equipped with the product Borel $\sigma$-algebra $\bigotimes_{x\in\X}\mathcal{B}_{\mathbb{R}}$ and the product probability measure 
$\mathbb{P}:=\bigotimes_{x\in\X}\mathbb{P}_0$.
We assume throughout that the single-site distribution $\mathbb{P}_0$ is compactly supported, $\supp\mathbb{P}_0\subseteq[v_{-},v_{+}]$ for some $v_{-},v_{+}\in\mathbb{R}$, in order to avoid irrelevant technical complications in dealing with unbounded operators. For the proof of the Lifshits tail in 
Theorem~\ref{lif-thm} we suppose in addition that $\mathbb{P}_0$ is not concentrated at one single point, i.e.\ 
\begin{equation} \label{Ass1}
\mathbb{P}_0(\{v\})<1 \quad \text{~for every~} v\in \mathbb{R},
\end{equation}
and that its upper tail decays no faster than any power, i.e.\ there exist real constants $C,\mu >0$ such that 
\begin{equation} \label{Ass2} 
 \mathbb{P}_0(\left[v_{+}-\varepsilon,v_{+} \right]) \geq C\varepsilon^\mu
\end{equation}
for every $\varepsilon>0$ sufficiently small.

The operator $\Delta$ in \eqref{Hamiltonian} is the hierarchical Laplacian \eqref{Lapl} and refers to a \emph{hierarchical structure} on $\X$, which we need to explain first. A hierarchical structure on $\X$ is a sequence of partitions $(\mathcal{P}_r)_{r\in\mathbb{N}_{0}}$ of $\X$ together with a sequence  $(n_{r})_{r\in\mathbb{N}}$ of natural numbers such that the properties (H1) -- (H3) below hold. By definition, each partition subdivides $\X$ into mutually disjoint subsets, which we call \emph{clusters}. The clusters of $\mathcal{P}_r$ are referred to as \emph{clusters of rank} $r$.
\begin{enumerate}
 \item[(H1)] $\mathcal{P}_{0}$ is the trivial partition, the clusters of which consist precisely of the single elements of $\X$.
 \item[(H2)] Every cluster of rank $r \in\mathbb{N}$ is a union of $n_{r}$ distinct clusters of rank $r-1$.
 \item[(H3)] Given $x,y \in\X$ there is a cluster of some rank containing both $x$ and $y$.
\end{enumerate}

\medskip
\noindent
We denote by $Q_r(x)$ the unique cluster of rank $r$ containing $x\in \X$, and we write $|A|$  for the number of elements of a finite set $A$. By (H2) the number of elements $|Q_{r}(x)| = \prod_{r'=1}^{r} n_{r'}$ of this cluster does not depend on $x\in\X$. Thus, we will simply write $|Q_{r}|$ for the cluster size. The elements of $\X$ can be enumerated 
\begin{equation}
	\label{enumeration}
	\mathbb{N}_{0} \rightarrow \X,  \qquad k \mapsto x_{k},
\end{equation}
in such a way that $x_{k_{1}}$ and $x_{k_{2}}$ belong to the same cluster of rank $r$ if and only if there exists $M\in\mathbb{N}_{0}$ with $ k_{1}, k_{2} \in \{M |Q_{r}|,\ldots, (M+1)|Q_{r}| -1\}$, see Fig.~\ref{fighs}.

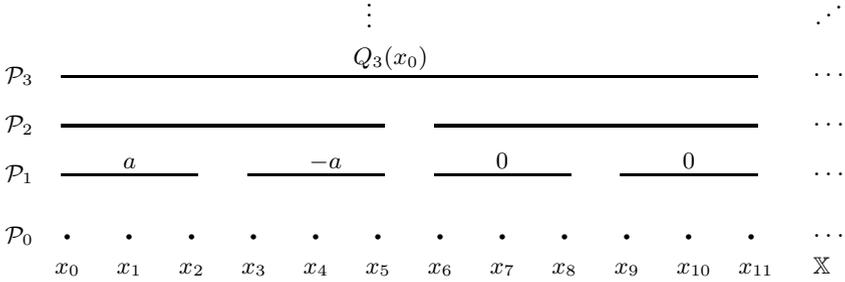
\begin{figure}[t]
	\setlength{\unitlength}{0.014\textwidth}
\begin{picture}(65,26)(-4,-4)
	\multiput(0,0)(5,0){12}{\circle*{0.5}} %\put(60,0){\circle*{0.5}} 
	\put(60,0){$\ldots$}
	%---Beschriftung--- 
	\put(-1,-3){\small $x_0$}
	\put(4,-3){\small $x_1$}
	\put(9,-3){\small $x_2$}
	\put(14,-3){\small $x_3$}
	\put(19,-3){\small $x_4$}
	\put(24,-3){\small $x_5$}
	\put(29,-3){\small $x_6$}
	\put(34,-3){\small $x_7$}
	\put(39,-3){\small $x_8$}
	\put(44,-3){\small $x_9$}
	\put(49,-3){\small $x_{10}$}
	\put(54,-3){\small $x_{11}$}
	\put(60,-3){\small $\mathbb{X}$}
	%\put(34,-3){$\ldots$}
	\thicklines
	\multiput(-0.5,5)(15,0){4}{\line(1,0){11}} \put(60,5){$\ldots$} 
	\put(4.5,5.5){\small $a$} \put(19.5,5.5){\small $-a$}
	\put(34.5,5.5){\small $0$} \put(49.5,5.5){\small $0$}
	\multiput(-0.5,9)(30,0){2}{\line(1,0){26}}\put(60,9){$\ldots$} 
	%\put(8,9.5){\small $a+b=0$}
	\put(-0.5,13){\line(1,0){56}} \put(60,13){$\ldots$} 
	\put(23,14){\small $Q_3(x_0)$} \put(24,17){$\vdots$} \put(60,17){$\diagdots$}
	%---Beschriftung---
	\put(-5,-0.5){\small $\mathcal{P}_0$}
	\put(-5,4.5){\small $\mathcal{P}_1$}
	\put(-5,8.5){\small $\mathcal{P}_2$}
	\put(-5,12.5){\small $\mathcal{P}_3$}
\end{picture}
	\caption{Sketch of a hierarchical structure for $n_1=3, n_2=2, n_3=2$ up to the 
		rank-$3$-cluster containing $x_0$.
		A function taking on the values $\pm a\in\mathbb{C}\setminus \{0\}$, when restricted to the 
		clusters $Q_1(x_0), Q_1(x_3)$, and being equal to zero everywhere else is an 
		eigenfunction of the hierarchical Laplacian $\Delta$ corresponding to the eigenvalue $\lambda_1$.} 
	\label{fighs}
\end{figure}

A hierarchical structure is called \emph{homogeneous of degree $n\in\mathbb{N}$}, if 
\begin{equation}\label{homogeneous}n_r=n \quad \text{for all } r\in\mathbb{N}.
\end{equation}
In this case, the size of any cluster of rank $r$ is given by $|Q_r|=n^r$.

Given a sequence of probability weights $(p_s)_{s\in\mathbb{N}}$, $0 < p_{s}< 1$ for all $s\in\mathbb{N}$ and $\sum_{s=1}^\infty\,p_s=1$, the \emph{hierarchical Laplacian} is defined as the weighted sum 
\begin{equation}
	\label{Lapl}
 	\Delta:= \sum_{s=1}^\infty \, p_s\, \E_s 
\end{equation}
of the \emph{cluster averaging operators} $\E_s\colon \ell^2(\X)\rightarrow \ell^2(\X)$, 
\begin{equation}
	(\E_s\psi)(x) := \frac{1}{|Q_s|}\sum_{y\in Q_s(x)}\psi(y),
\end{equation}
where $\psi\in \ell^2(\X)$, $x\in\X$ and $s\in\mathbb{N}_{0}$.
For convenience, we also introduce $p_0:=0$. Note that $\E_{0} = \one$, the identity operator, and that $p_{s} \neq 0$ for every $s\in\mathbb{N}$.
The random operator \eqref{Hamiltonian} is referred to as the \emph{hierarchical Anderson Hamiltonian} or the \emph{\HAM}. We use the additional specification ``homogeneous'' if \eqref{homogeneous} holds.

The basic spectral theorem for the \sa{} \hl{} $\Delta$ is

\begin{theorem}
	The spectral decomposition of the \hl{} $\Delta$ reads
	\begin{equation} 
		\label{Decomp}
		\Delta = \sum_{r=0}^{\infty} \lambda_r \left( \E_r - \E_{r+1} \right),
	\end{equation}
	where 
	\begin{equation}
		\lambda_r := \sum_{s=0}^r\,p_s, \quad r\in \mathbb{N}_{0}, \label{eigenvalues}
	\end{equation} 
	are its eigenvalues (of infinite multiplicity) and $\E_r - \E_{r+1}$ is the orthogonal projection onto 	
	the eigenspace corresponding to $\lambda_r$. 
	In particular, these eigenvalues and their accumulation point 
	$\lambda_\infty:=1$ belong to the essential spectrum of $\Delta$.
\end{theorem}

We refer to \cite[Thm.~1.1]{K1} for a proof. The spectral value $\lambda_{\infty}=1$ is never an eigenvalue since $p_s>0$ for every $s\in\mathbb{N}$.
Any eigenfunction $\psi_r\in \ell^2(\X)$ corresponding to the eigenvalue $\lambda_r$ is constant on every cluster of rank $r$. In addition, the sum of these constants over all rank-$r$-clusters which belong to the same rank-$(r+1)$-cluster is always equal to zero, see Fig.~\ref{fighs} for a sketch. 

Using \cite[Lemma~1.2]{K2} together with ergodicity of $\Ha^{\omega}$, see Lemmas~\ref{ergodic-shifts} 
and~\ref{covariant}, we conclude the standard

\begin{lemma}
 	\label{det-spec}
	There exists a non-random compact subset $\Sigma\subset \mathbb{R}$ such that 
	$\spec\Ha^{\omega} = \Sigma$ 
	for $\mathbb{P}$-a.e.\ $\omega\in\Omega$. The location of the deterministic spectrum $\Sigma$ obeys 
	\begin{equation}
		\label{spectrum}
		 \spec\Delta + \supp\mathbb{P}_0 \subseteq \Sigma 
		\subseteq \big( \spec\Delta + \cohu (\supp\mathbb{P}_0 )\big)
		\cap \big( [0,1] + \supp\mathbb{P}_0 \big),
	\end{equation}
	where $\cohu B$ denotes the convex hull of a set $B\subset \mathbb{R}$. In particular, we have 
	$\sup (\inf) \Sigma = \sup(\inf) \{\spec\Delta + \supp\mathbb{P}_0\}$, 
	and if $\supp\mathbb{P}_0$ is even connected, then also $\Sigma = \spec\Delta + \supp\mathbb{P}_0$.
\end{lemma}

\begin{remark}
The preceding lemma strengthens Lemma~1.2 in \cite{K2}, in as much as non-randomness of $\spec\Ha^{\omega}$ is established irrespective of the connectedness of $\supp\mathbb{P}_0$.
 \end{remark}

\noindent
For homogeneous hierarchical structures, we will focus on the special case where the decay rate of $(p_s)_{s\in\mathbb{N}}$ is linked to the degree $n$ of the structure.

\begin{definition}
	\label{d-s-def}
	Consider a homogeneous hierarchical structure of degree $n\geq 2$. 
	Suppose that there exist constants $C_1, C_2>0$ and $\rho>1$ such that
	\begin{equation}
		\label{algdec}
		C_1\rho^{-r}\leq p_r \leq C_2\rho^{-r} 
	\end{equation}
	for all $r\in\mathbb{N}$ large enough. Then the \emph{spectral dimension} of this model is defined as
	\begin{equation}
		\label{spectrdim}
		d_s\equiv d_s(n, \rho):= 2 \;\frac{\ln n}{\ln \rho}\,.
	\end{equation}
\end{definition}

\medskip
In other words, this amounts to $ C_{1} n^{-2r/d_s} \le p_r \le C_{2} n^{-2r/d_s}$ for large $r$.
One motivation for the definition of $d_s$ will be given by Lemma \ref{vanHove} below. To this end 
we introduce the \emph{integrated density of states} 
$\N_{0}: \mathbb{R} \to [0,1]$ of $\Delta$, which is defined by 
\begin{equation}
 	E \mapsto \N_0(E) := \langle\delta_{x_0},\Chi_{]-\infty, E]}(\Delta) \delta_{x_0} \rangle
\end{equation}
for some $x_0 \in\X$. Here, $\Chi_B$ stands for the indicator function of a set $B$, $\langle\cdot,\cdot\rangle$ denotes the canonical scalar product of the Hilbert space $\ell^2(\X)$ and $\delta_{x}$ the canonical basis vector associated with $x\in \X$, i.e.\ $\delta_{x}(y) =0$ for every $y\in\X\setminus\{x\}$ and $\delta_{x}(x)=1$. We remark that $\N_0$ is a right-continuous distribution function, which is normalised according to $\N_{0}(1)=1$. Moreover, it does not depend on the choice of $x_0$. This can be seen from the more explicit expression
\begin{equation}
	\label{explicit}
 	N_{0}(E) =\sum_{r\in\mathbb{N}_{0}:\: \lambda_{r} \le E} \bigg( \frac{1}{|Q_{r}|} - 
 	\frac{1}{|Q_{r+1}|} \bigg) = 1- \frac{1}{|Q_{r(E) + 1}|},
\end{equation}
which follows from the spectral representation \eqref{Decomp}. The second equality in \eqref{explicit} makes only sense for $E\in [0,1[$ with $r(E) := \max\{ r' \in \mathbb{N}_{0}: \lambda_{r'} \le E\}$.

\begin{lemma}
	\label{vanHove}
	Let $\Delta$ be the hierarchical Laplacian of a homogeneous hierarchical structure with spectral dimension $d_s$. Then the integrated density of states $N_{0}$ of $\Delta$ exhibits the upper-edge asymptotics
	\begin{equation} 
		\lim_{E\searrow 0}\frac{\ln\left[ 1 - \N_{0}(1-E)  \right]}{\ln E}= \frac{d_s}{2}\,.
	\end{equation} 
\end{lemma}

\medskip
\noindent
We refer to, e.g., \cite[Prop.~1.3]{K1} or \cite{M} for a proof.

Next, we turn to the central quantity of this paper. The \emph{integrated density of states} $\N: \mathbb{R} \to [0,1]$ of the \HAM{} is defined by
\begin{equation}
	\label{dosmeasure}
	E \mapsto \N(E) :=\mathbb{E}\bigl[ \langle\delta_{x_0}, \Chi_{]-\infty, E]} 
	(\Ha^\omega)\delta_{x_0} \rangle\bigr],
\end{equation}
where $\mathbb{E}$ denotes the probabilistic expectation associated with $\mathbb{P}$. 
The integrated density of states is a right-continuous distribution function and independent of the choice of $x_0\in\X$.
The set of growth points of $\N$ coincides $\mathbb{P}$-a.s.\ with the compact deterministic spectrum described in \eqref{spectrum}. At the lower spectral edge $\inf \supp\mathbb{P}_0$ the asymptotics of $\N$ is solely determined by the single-site distribution $\mathbb{P}_0$ of the potential, because the second-lowest eigenvalue $\lambda_1=p_1>0$ of the \hl{} is separated from $\lambda_0=0$ by a gap. The interesting case is the upper spectral edge $1+ \sup\supp \mathbb{P}_{0}$, because the eigenvalues of the \hl{} accumulate at $\sup\spec\Delta =1$. The main result of the present paper, Theorem~\ref{lif-thm}, concerns
this case.

\begin{theorem}
	\label{lif-thm} 
	Let $N$ be the integrated density of states of a homogeneous \HAM{} with the properties \eqref{Ass1}, \eqref{Ass2}  and \eqref{algdec}. Then, $N$ has a Lifshits tail at the upper spectral edge $1+v_{+}$ in 
	the sense that
	\begin{equation}
		\label{LTeq} 
		\lim_{E\searrow 0} \frac{ \ln \big|\ln \big[1 - \N(1+v_{+}-E) \big]\big|}{\ln E}= - \frac{d_s}{2}\,.
	\end{equation}
	We note that $\N(1+v_{+}) =1$ and recall that $d_s$ denotes the spectral dimension \eqref{spectrdim} of the homogeneous hierarchical model.
\end{theorem} 

\noindent
The proof of Theorem~\ref{lif-thm} can be found in Sect.~\ref{LT}. 

\begin{remark} 
	\label{strategy}
	The intuition behind Theorem~\ref{lif-thm} is that eigenvalues $1+v_{+}-E$ close to the upper edge of the spectrum, i.e.\ for $E \ll 1$,  must maximise both the kinetic and the potential energy. In order to achieve a large kinetic energy of the order $1-E =: 1 - \sum_{r= \kappa(E)}^{\infty} p_{r}$, the associated  eigenfunction must be approximately constant over a cluster of high rank $\kappa(E) \sim - d_{s} (\ln E)/(2 \ln n) \gg 1$, that is, with volume $|Q_{\kappa(E)}| = n^{\kappa(E)} \sim E^{-d_{s}/2}$. To ensure a large potential energy at the same time then requires most coupling constants of this cluster to take values close to $v_{+}=\sup\supp \mathbb{P}_{0}$. This is a large-deviation event with approximate probability $\exp\{- |Q_{\kappa(E)}|\} \sim \exp\{ -E^{-d_{s}/2}\}$ that sets the scale for $N(1+v_{+}-E)$.		We refer to the end of Sect.~\ref{intro} for further comments on the fact that the Lifshits exponent equals $d_{s}/2$. 
\end{remark}

%%%%%%%%%%%%%%%%%%%%%%%%%%%%%%%%%%%%%%%%%%%%%%%%%%%%%%%%%%%%%%%%%%%%%%%%%%%
%%%%%%%%%%%%%%%%%%%%%%%%%%%%%%%%%%%%%%%%%%%%%%%%%%%%%%%%%%%%%%%%%%%%%%%%%%%

\section{Dirichlet-Neumann bracketing}\label{Sandwich}

Technically, many proofs of \LT{} rely on Dirichlet-Neumann bracketing. This allows for a two-sided estimate of the integrated density of states in terms of finite-volume operators. We will also follow this route.
Thus, it is a main point of this paper to find a pair of suitable finite-volume restrictions of the hierarchical Laplacian for which Dirichlet-Neumann bracketing works. 

\begin{definition}
	For $x_{0} \in\X$ fixed and $\kappa \in\mathbb{N}_{0}$ we consider the finite cluster $Q_{\kappa} \equiv Q_\kappa(x_0)$ and introduce the \emph{Neumann}, resp.\ \emph{Dirichlet finite-volume restrictions}
	\begin{equation}
		\Delta_{\mathrm{N}, Q_{\kappa}} :=  \sum_{s=1}^\kappa\,p_s \E_s\Big|_{\ell^2(\Q)} , \qquad
		\Delta_{\mathrm{D}, Q_{\kappa}} :=  \Delta_{\mathrm{N}, Q_{\kappa}} +  \sum_{s=\kappa+1}^\infty\, p_s \, \one\Big|_{\ell^2(\Q)}
	\end{equation}
	of the hierarchical Laplacian to the finite-dimensional subspace $\ell^2(Q_\kappa)$. For $\mathrm{X} \in \{\mathrm{N}, \mathrm{D}\}$ we then set
	\begin{equation} 
		\Ha_{\mathrm{X},Q_{\kappa}}^{\omega} :=\Delta_{\mathrm{X}, Q_{\kappa}} +V^\omega.
	\end{equation}
To simplify notation we write $\Delta_{\mathrm{X}, \kappa} \equiv \Delta_{\mathrm{X}, Q_{\kappa}}$ and  
$\Ha_{\mathrm{X},\kappa}^{\omega} \equiv \Ha_{\mathrm{X},Q_{\kappa}}^{\omega}$, if there is no danger of confusion.
\end{definition}

\noindent
The desired property is stated in 

\begin{lemma}[Dirichlet-Neumann decoupling]
	\label{DNcoupl}
	Consider a fixed finite cluster $Q_{\kappa}$ of rank $\kappa\in\mathbb{N}$. Let $r \in\mathbb{N}_{0}$, $r 
	< \kappa$, and assume that the cluster $Q_{\kappa}$ is the union of $m$ disjoint clusters $Q_r^1, 
	\ldots, Q_r^m$ of lower rank $r$.
	Writing $\Ha_{\Xr}^{\omega,j} \equiv \Ha_{\mathrm{X}, \smash{Q_{r}^{j}}}^{\omega\phantom{j}}$ for $\mathrm{X}\in\{\mathrm{N},\mathrm{D}\}$, we have in the 
	sense of quadratic forms
	\begin{equation}
		\Han \geq \bigoplus_{j=1}^m \Ha_{\Nr}^{\omega,j} \quad \text{and} \quad 
		\Had \leq \bigoplus_{j=1}^m \Ha_{\Dr}^{\omega,j}\,.
	\end{equation}
\end{lemma}

\begin{proof}
The subspace $\ell^2(Q_r^j)$ is left invariant under $\E_{s} \big|_{\ell^{2}(Q_{\kappa})}$ for every $s \in \{1, \ldots, r\}$ and every $j\in \{1,\ldots, m\}$. Thus, we have
\begin{align} 
	\Han &= \biggl( \bigoplus_{j=1}^m \Ha_{\Nr}^{\omega, j} \biggr) + \sum_{s=r+1}^\kappa \, 
	p_s\E_s \Big|_{\ell^{2}(Q_{\kappa})}  
	\geq \bigoplus_{j=1}^m \Ha_{\Nr}^{\omega, j} \\
\intertext{and}
\Had &= \biggl( \bigoplus_{j=1}^m \Ha_{\Dr}^{\omega, j} \biggr) -  
	\sum_{s=r+1}^\kappa\,p_s \big(\one - \E_s\big) \Big|_{\ell^{2}(Q_{\kappa})}
\, \leq \, \bigoplus_{j=1}^m \Ha_{\Dr}^{\omega,j}.
\end{align} 
\end{proof}

\noindent
Being an operator on a finite dimensional space, $\Ha_{\Xk}^{\omega}$ has discrete (random) eigenvalues
\begin{equation} 
	\label{e-not}
	e_{\Xk}^{\omega}(1) \leq e_{\Xk}^{\omega}(2) \leq \ldots \leq 
	e_{\Xk}^{\omega}(|\Q|),
\end{equation} 
which are counted according to their multiplicities.
We define the corresponding normalised eigenvalue counting function $\N_{\Xk}^{\omega}
\colon\mathbb{R}\to [0,1]$ by
\begin{equation} 
	\label{rdos}
	E \mapsto \N_{\Xk}^{\omega}(E):= \frac{1}{|Q_\kappa|} \sum_{j=1}^{|\Q|}
	\Chi_{]-\infty, E]} \big(e_{\Xk}^{\omega}(j)\big)\,, \quad \mathrm{X} \in\{\mathrm{N},\mathrm{D}\}.
\end{equation}
In the macroscopic limit this quantity is self-averaging and justifies the interpretation of $\N$ as an integrated density of states.
 
\begin{lemma}
	\label{conv}
	For $\mathrm{X}\in\{\mathrm{N},\mathrm{D}\}$ there exists a set $\Omega_{0} \subseteq \Omega$ of full probability, $\PP[\Omega_{0}] =1$, such that for every $\omega \in\Omega_{0}$ we have 
	\begin{equation}
		\label{convN}
		\lim_{\kappa\to\infty} \N_{\Xk}^{\omega}(E) = \N(E) 
	\end{equation}
	at each continuity point $E$ of $N$. 
	%Moreover for $\mathrm{X}=\mathrm{N}$, \eqref{convN} holds for all $E \in\mathbb{R}$. 
\end{lemma}

\begin{remark}
  The above lemma extends Thm.~3.3 in \cite{K3} to the Dirichlet case.
  % and to discontinuity points (Neumann case only).
  Whereas we rely on ergodicity, the (longer) argument in \cite{K3} follows a different route which is based on the law of large numbers instead. 
\end{remark}

%\begin{lemma}
%	\label{conv}
%	For $\mathrm{X}\in\{\mathrm{N},\mathrm{D}\}$ and for $\mathbb{P}$-a.e.\ $\omega\in\Omega$ we have that
%	\begin{equation}
%		\label{convN}
%		\lim_{\kappa\to\infty} \N_{\Xk}^{\omega}(E) = \N(E) \qquad \text{for all~~}	 E\in\mathbb{R}.
%	\end{equation}
%\end{lemma}
%
%\begin{remark}
%  The above lemma is slightly stronger than Thm.~3.3 in \cite{K3}, since it also establishes 
%  convergence at discontinuity points of $N$. 
%\end{remark}

\begin{proof}[Proof of Lemma~\ref{conv}]
\emph{Case $\mathrm{X} =\mathrm{N}$.} 
Given $Q \subseteq \X$, we write $\tr_{Q}$ for the trace over $\ell^{2}(Q)$ and identify a function on $Q$ with the corresponding multiplication operator by this function on $\ell^{2}(Q)$. The covariance and ergodicity of the \HAM{}, Lemmas~\ref{covariant} 
and~\ref{birk}, yield  for every continuous function $\varphi \in C_{c}(\mathbb{R})$ with compact support
\begin{align}
 	\frac{1}{|\Q|} \tr_{\X} \big[ \Chi_{\Q(x_{0})} \, \varphi (\Ha^{\omega})\big] &= 
		\frac{1}{|\Q|} \sum_{x\in\Q(x_{0})} \big\langle \delta_{x_{0}},  
		\varphi \big(\Ha^{\tau_{x}(\omega)}\big)\delta_{x_{0}}\big\rangle \nonumber\\
		&\stackrel{\kappa\to\infty}{\relbar\joinrel\longrightarrow} \; \int_{\Omega} \mathrm{d}\kern1pt\PP(\omega') \; 
			 \langle \delta_{x_{0}},  \varphi (\Ha^{\omega'})\delta_{x_{0}}\rangle 
\end{align}
for $\PP$-a.e.\ $\omega\in\Omega$. Next, we introduce the unrestricted but truncated operator $K^{\omega}_{\kappa} :=\sum_{s=1}^{\kappa} p_{s} \E_{s} + V^{\omega}$ on $\ell^{2}(\X)$. 
Since $K^{\omega}_{\kappa}$ converges to $H^{\omega}$ in operator norm as $\kappa\to\infty$ for $\PP$-a.e.\ $\omega\in\Omega$, we conclude that also 
\begin{equation}
 	\frac{1}{|\Q|} \tr_{\X} \big[ \Chi_{\Q(x_{0})} \, \varphi (K_{\kappa}^{\omega})\big] \stackrel{\kappa\to\infty}{\relbar\joinrel\longrightarrow} \; \int_{\Omega} \mathrm{d}\kern1pt\PP(\omega') \; 
			\langle \delta_{x_{0}},  \varphi (\Ha^{\omega'})\delta_{x_{0}}\rangle
\end{equation}
for $\PP$-a.e.\ $\omega\in\Omega$ and every given $\varphi \in C_{c}(\mathbb{R})$. This implies $\mathbb{P}$-a.s.\ vague convergence of the corresponding probability measures on $\mathbb{R}$, see e.g.\ the proof of Thm.~5.5 in \cite{Kirsch}. Hence, there exists a set $\Omega_{0} \subseteq \Omega$ of full probability, $\PP[\Omega_{0}] =1$, such that for every $\omega \in\Omega_{0}$ the corresponding distribution functions converge
\begin{equation}
	\label{relate}
	 \N_{\Nk}^{\omega}(E) = \frac{1}{|\Q|} \tr_{\X} 	\big[ \Chi_{\Q(x_{0})} \, 
		  \Chi_{]-\infty, E]} (K^{\omega}_{\kappa})\big] 
		\stackrel{\kappa\to\infty}{\relbar\joinrel\longrightarrow} \; N(E)
\end{equation}
for every continuity point $E\in\mathbb{R}$ of $N$. 
We note that the left equality above relies on $\ell^{2}(\Q(x_{0}))$ being an invariant subspace of $K_{\kappa}^{\omega}$.

\emph{Case $\mathrm{X} =\mathrm{D}$.} 
We repeat the argument of the Neumann case with 
\begin{equation}
	K^{\omega}_{\kappa} :=\sum_{s=1}^{\kappa} p_{s} \E_{s} +  \sum_{s=\kappa+1}^{\infty} p_{s}\, \one 
	+ V^{\omega} 
\end{equation}
acting on $\ell^{2}(\X)$.
\end{proof}

\begin{lemma}[Dirichlet-Neumann bracketing]
	\label{sandwich}
	For every cluster $Q_r$ of rank $r\in\mathbb{N}_0$ and for every $E\in\mathbb{R}$ the integrated density of states $\N$ obeys the two-sided estimate
	\begin{equation} 
		\label{sandwich-eq}
		\mathbb{E}\bigl[\N^{\omega}_{\Dr}(E)\bigr] \leq \N(E) \leq 
		\mathbb{E}\bigl[\N^{\omega}_{\Nr}(E)\bigr] .
	\end{equation}
\end{lemma}

\begin{proof}
We fix $E \in\mathbb{R}$, $r\in\mathbb{N}_{0}$ and $x_{0} \in\X$. Consider a cluster $\Q\equiv\Q(x_0)$ of rank $\kappa\in\mathbb{N}$, $\kappa > r$, such that $\Q$ is the union of $m$ disjoint rank-$r$-clusters $Q_r^1, 
	\ldots, Q_r^m$ for some $m\in\mathbb{N}$.
Using Lemma \ref{DNcoupl} and the min-max principle, we conclude that
\begin{align}\label{estimate}
	\mathbb{E}\biggl[\frac{1}{|Q_\kappa|}\text{tr}_{\Q} \Chi_{]-\infty, E]} ( \Han) \biggr] 
	&\leq \mathbb{E}\biggl[ \frac{1}{m|Q_r^1|} \tr_{\Q} \biggl( \bigoplus_{j=1}^m 
	\Chi_{]-\infty, E]}  \big(\Ha_{\Nr}^{\omega,j}\big)  \biggr) \biggr] \nn\\ 
	&= \frac{1}{m|Q_r^1|} \; \sum_{j=1}^m \mathbb{E} \Bigl[  \tr_{Q_r^{j}}
		\Chi_{]-\infty, E]} \bigl( \Ha_{\Nr}^{\omega,j} \bigr)   \Bigr]  \nn \\ 
	&= \frac{1}{|Q_r^1|}  \; \mathbb{E}\Bigl[ \text{tr}_{Q_r^1} 
		\Chi_{]-\infty, E]} \bigl( \Ha_{\Nr}^{\omega,1} \bigr) \Bigr].
\end{align}
Here, the last equality relies on the $r$-cluster permutation invariance of $\Delta_{\Nr}$ and the identical distribution of the random variables. By 
dominated convergence and Lemma~\ref{conv}, the left-hand side of \eqref{estimate} converges to 
$\mathbb{E}[ \N(E)] = \N(E)$ as $\kappa\to\infty$, provided $E$ is a continuity point of $N$. In this case
we obtain $\N(E)\leq \mathbb{E}[\N_{\Nr}^{\omega}(E)]$. If $E$ happens to be a discontinuity point of $N$ -- for which we have no convergence statement in Lemma~\ref{conv} -- simply replace $E$ by a monotone decreasing sequence of continuity points $E_{l} \searrow E$ and use right-continuity. 

The lower bound follows by the same line of reasoning.
\end{proof}

%%%%%%%%%%%%%%%%%%%%%%%%%%%%%%%%%%%%%%%%%%%%%%%%%%%%%%%%%%%%%%%%%%%%%%%%%%%%%%%%
%%%%%%%%%%%%%%%%%%%%%%%%%%%%%%%%%%%%%%%%%%%%%%%%%%%%%%%%%%%%%%%%%%%%%%%%%%%%%%%%

\section{Proof of Theorem~\ref{lif-thm}} \label{LT}
The strategy outlined in Remark \ref{strategy} suggests to estimate the maximal eigenvalue of random operators on finite clusters.
For a general \sa{} operator, an upper bound on the maximal eigenvalue is provided by Temple's inequality, which we recall from \cite[Thm.\ XIII.5]{RS4} (with $A$ replaced by $-A$).

\begin{lemma}\label{temple}
	Let A be a \sa{} operator in a Hilbert space and let $E_{\mathrm{max}}(A) := \sup \spec A$ 
	be an isolated eigenvalue of $A$. 
	We write $E_1(A):= \sup\big\{ \spec (A) \setminus \{E_{\mathrm{max}}(A)\}\big\}$ and assume the existence of 
	a vector $\psi$ in the domain of $A$ such that $\langle\psi,\psi\rangle=1$  and
	\begin{equation}
		\label{precond} 
		\langle \psi, A\psi \rangle  >  E_1(A).
	\end{equation} 
	Then we have the estimate 
	\begin{equation} 
		\label{temple-eq}
		E_{\mathrm{max}}(A)  \leq  \langle \psi,A\psi\rangle + \frac{\langle \psi, A^2\psi\rangle
		- \langle \psi,A\psi\rangle^2}{\langle\psi, A\psi\rangle -E_1(A)}. 
	\end{equation}
\end{lemma}

\medskip 
\begin{proof}[Proof of Theorem \ref{lif-thm}]
We assume w.l.o.g.\ that $v_{+} =0$. This can always be achieved by adding the constant term $-v_{+} \one$ to $\Ha^{\omega}$. 
In order to prove the assertion, we will construct an upper and a lower bound on $1- \N(1-E)$, which, after taking double logarithms, will asymptotically coincide as $E\searrow 0$. 

In what follows we choose some fixed cluster $\Q \equiv Q_\kappa(x_0),\, x_0 \in \X$, of finite rank $\kappa\in\mathbb{N}$ and let $0<E\ll 1$ be arbitrary but fixed. 

(a) \quad \emph{Upper bound.}\quad We recall the notation \eqref{e-not} and estimate the expectation of the finite-volume eigenvalue counting function $N_{\Dk}^{\omega}(E')$ for every $E'\in\mathbb{R}$ in terms of the maximal eigenvalue $\Emd \equiv e_{\Dk}^{\omega}(|\Q|)$ of $\Had$ by 
\begin{align}
	\label{formereq} 
	\mathbb{E}\left[ \N_{\Dk}^{\omega}(E')\right] 
	&= \frac{1}{|\Q|} \; \mathbb{E} \biggl[ \sum_{j=1}^{|\Q|} 
		\Chi_{]-\infty, E']} \big(e_{\Dk}^{\omega}(j) \big)\biggr]  \nn\\ 
	& \geq \mathbb{E} \Bigl[ \Chi_{]-\infty, E']} \big( \Emd \bigr) \Bigr] 
	  = \mathbb{P} \bigl[ \Emd \leq E' \bigr]  .
\end{align}
Setting $E' =1-E$ and applying Lemma~\ref{sandwich}, we conclude that
\begin{equation}
	\label{r1}
	1-\N(1-E) \leq 1- \mathbb{E}\left[ \N_{\Dk}^{\omega}(1-E)\right] 
	\leq \mathbb{P} \big[ E_{\mathrm{max}}(\Ha_{\Dk}^{\omega}) > 1-E \big]. 
\end{equation}
To proceed further, we need an upper bound on $\Emd$. Temple's inequality cannot be applied directly: even the normalised trial function $\psi_0=|\Q|^{-1/2} \in \ell^{2}(\Q)$, which is the eigenfunction of $\Delta_{\Dk}$ corresponding to the maximal eigenvalue $E_{\mathrm{max}}(\Delta_{\Dk}) = 1$, will not satisfy the condition \eqref{precond} for $\kappa$ large enough. This problem is circumvented by introducing the auxiliary operator 
\begin{equation}
	 \Hd := \Delta_{\Dk} + V^\omega_{\kappa} \ge  \Delta_{\Dk} + V^{\omega} = \Had 
\end{equation}
on $\ell^{2}(\Q)$ with the new, higher potential $V^\omega_{\kappa}(x):= 
\max\left\{ \omega_{x}, - p_\kappa/3 \right\} \le 0$ for all $x\in\X$. This operator satisfies
\begin{equation}
	\label{calcabs}
	\langle \psi_{0}, \Hd\psi_{0} \rangle  \geq 1-\frac{p_\kappa}{3}
	> 1-p_\kappa = E_1 (\Delta_{\Dk}) \geq E_1(\Hd)
\end{equation}
and may thus be employed in Temple's inequality.
In order to simplify the right-hand side of \eqref{temple-eq}, we use the estimates $\langle \psi_{0}, \Hd\psi_{0} \rangle -  E_1(\Hd) \ge 2p_{\kappa}/3$ and 
$(V^\omega_{\kappa})^2 \leq (p_\kappa/3)|V^\omega_{\kappa}|$, and  arrive at 
\begin{equation}
	\Emd \leq 1 + \frac{1}{2|\Q|} \sum_{x\in\Q} V^\omega_{\kappa}(x).
\end{equation}
Together with \eqref{r1}, this implies
\begin{equation}
	\label{r2}
	1-\N(1-E) \leq \mathbb{P}\biggl[ \frac{1}{|\Q|} \sum_{x\in\Q} V^{\omega}_{\kappa}(x) > -2E \biggr].
\end{equation}
By hypothesis of Theorem~\ref{lif-thm}, the spectral dimension $d_{s}$ exists and the estimate \eqref{algdec} is valid. So we choose 
\begin{equation}
	\label{choice} 
	\kappa =  k(E) := \max\bigl\{r\in\mathbb{N}:|Q_r|\leq (\alpha E)^{- d_s/2}\bigr\}, 
\end{equation}
where $\alpha>0$ is a free parameter to be determined below. Since, by definition, $n=\rho^{d_{s}/2}$ for some $\rho >1$ and since $|Q_{r}|=n^{r}$, the inequality in \eqref{choice} is equivalent to 
$\rho^{-r} \ge \alpha E$. Therefore, \eqref{algdec} guarantees that for $E>0$ small 
enough, we have the estimate $p_{ k(E)} > C_1 \rho^{- k(E)} \geq C_1 E\alpha$ with some constant $C_{1}>0$. This estimate and \eqref{r2} imply
\begin{equation}
	\label{probestimate} 
	1-\N(1-E) \leq \mathbb{P}\biggl[ \frac{1}{|Q_{ k(E)}|} \sum_{x\in Q_{ k(E)}}  
		V^{\omega}_{ k(E)}(x) 	> - \frac{2p_{ k(E)}}{\alpha\,C_1}\biggr] 
		=: \mathbb{P} \big[\mathcal{A}_{ k(E)}\big]
\end{equation}
for all $E>0$ small enough.

For every $r\in\mathbb{N}$, and assuming $\alpha > 6/C_{1}$, it follows that the condition
\begin{equation}
	\frac{1}{|Q_{r}|} \, \Big| \Bigl\{ x\in Q_{r}: V^{\omega}_{r}(x) 
		> -\frac{p_r}{3} \Bigr\} \Big| \geq 1-\frac{6}{\alpha C_1} =:z >0
\end{equation}
is necessary for the event $\mathcal{A}_{r}$ to occur. Hence
\begin{equation} \label{r4} 
	 \mathbb{P} [\mathcal{A}_{r}]  \leq 
	 \mathbb{P} \biggl[ \frac{1}{|Q_{r}|} \,  \Big|\Bigl\{ x\in Q_{r}: V^{\omega}_{r}(x) > -\frac{p_r}{3} \Bigr\}  \Big|
	 \geq z\biggr]. 
\end{equation}
To eliminate the dependence of this probability on $p_r$, we pick 
$\gamma\in ] \inf \supp\mathbb{P}_0\,,0[$, introduce the i.i.d.\ random variables 
$\eta_x^{\omega}:=\Chi_{]\gamma,0]} (\omega_{x})$  for $x\in\X$ and note that
\begin{equation}
	 \Big|\Bigl\{ x\in Q_{r}: V^{\omega}_{r}(x) > -\frac{p_r}{3} \Bigr\} \Big|
	= \, \Big|\Bigl\{ x\in Q_{r}: \omega_{x} > -\frac{p_r}{3} \Bigr\}  \Big|
	\le \sum_{x\in Q_{r}} \eta_{x}^{\omega},
\end{equation}
where we assumed $\gamma \le - p_{r}/3$ in the last step. This implies
\begin{equation}
 	 \mathbb{P} [\mathcal{A}_{r}]  
	 \leq  \mathbb{P}\biggl[ \frac{1}{|Q_{r}|} \sum_{x\in Q_{r}} \eta_x^{\omega} \geq z \biggr]
\end{equation}
for all $r \in \mathbb{N}$ large enough, since $(p_{r})_{r}$ is a null sequence.
The latter probability can be handled with the standard estimate
$\mathbb{P}(X\geq\delta)\leq e^{-t\delta}\,\mathbb{E}(e^{tX})$ for any $t\geq 0$ and any random variable $X$. Thus there exists $r_{0}\in\mathbb{N}$ such that for every $r\in\mathbb{N}$ with $r \ge r_{0}$ we have 
\begin{equation}
	\mathbb{P} [\mathcal{A}_{r}]  
	\leq \e^{-tz|Q_{r}|}\, \mathbb{E} \big[ \e^{t \sum_{x\in Q_{r}} \eta_x^{\omega}} \big] 
	= \e^{-tz|Q_{r}|} \, \Big( \mathbb{E}_{0} \big[ \e^{t\eta_{x_0}^{\omega}} 
			\big] \Big)^{|Q_{r}|} 
	= \e^{-|Q_{r}|f(t)}
\end{equation}
for any $t\geq 0$, where $f(t):=tz-\ln\mathbb{E}_{0} [ \e^{t\eta_{x_0}^{\omega}}]$ and $\mathbb{E}_{0}$ is the expectation associated with the single-site distribution $\mathbb{P}_{0}$.
In view of \eqref{probestimate} this means that there exists $E_{u} >0$ such that for every $E \in]0, E_{u}]$ the estimate
\begin{equation}
 	\label{r6}
	1- N(1-E) \le \e^{-|Q_{ k(E)}| f(t)}
\end{equation}
holds for all $t \ge 0$.

Now we choose $\gamma$ close enough to $0$ such that $q:=\mathbb{E}_{0}(\eta_x^{\omega}) =\mathbb{P}_0(\, ]\gamma,0]\, ) \in \; ]0,1[$. This is always possible in view of conditions \eqref{Ass1} and \eqref{Ass2}. 
By adjusting the free parameter $\alpha> 6/C_{1}$ large enough, we ensure in addition that $q<z$. Then we have $f(0)=0$ and 
\begin{equation} 
	f'(0) = z - \frac{\mathbb{E}_{0}\big[\eta_{x_0}^{\omega} \e^{t\eta_{x_0}^{\omega}}\big]}{\mathbb{E}_{0}
		\big[\e^{t\eta_{x_0}^{\omega}}\big]}\Bigg|_{t=0}=z-q>0 .
\end{equation}
Hence, there is a $t_0 > 0$ such that $f(t_0) > 0$. We remark that neither $t_{0}$ nor $f(t_{0})$ depend on $E$. Definition \eqref{choice} implies 
\begin{equation}
	\label{impl} 
	|Q_{ k(E)}| = \frac{1}{n}|Q_{ k(E)+1}| \geq \frac{1}{n}(\alpha E)^{-d_s/2}
	\end{equation} 
for the homogeneous model. This estimate and \eqref{r6} then yield the desired upper bound 
\begin{equation} 
	\label{ub}
	1- N(1-E)  \leq \exp\big\{ - C_{u}\, E^{- d_s/2} \big\} 
\end{equation}
for all $E\in ]0, E_{u}]$ with the constant $C_{u} := f(t_0)\,\alpha^{- d_s/2} / n > 0$.

(b)\quad \emph{Lower bound.}\quad 
This time we use the upper bound of Lemma~\ref{sandwich} and estimate
\begin{align} \label{l1}
	1- \N(1-E) &\geq 1- \mathbb{E}\bigl[ \N_{\Nk}^{\omega}(1-E)\bigr] \nonumber\\
	& = \frac{1}{|\Q|} \; \mathbb{E} \Bigl[  \big|\big\{\text{eigenvalues of } \Han > 1-E 
		\big\} \big| \Bigr] \nonumber\\
	& \geq \frac{1}{|\Q|} \; \mathbb{P}\bigl[ E_{\text{max}}(\Han) > 1-E\bigr],
\end{align}
where $\Emn \equiv e_{\Nk}^{\omega}(|\Q|) = \sup_{0 \neq \varphi\in \ell^2(\Q)} 
\langle\varphi,\Han\varphi\rangle / \langle\varphi,\varphi\rangle$ denotes the maximal 
eigenvalue of $\Han$.
The choice $\varphi =\psi_0 = |\Q|^{-1/2}$ for the trial function yields
\begin{equation}
	\Emn\geq \sum_{s=1}^\kappa p_s + \frac{1}{|\Q|}\sum_{x\in\Q}\omega_x .  
\end{equation}
Therefore we get together with \eqref{l1}
\begin{equation} 
  1 - \N(1-E)  \geq \frac{1}{|\Q|} \; \mathbb{P}\biggl[  \frac{1}{|\Q|}\sum_{x\in\Q}\omega_x  
  	> -E + \sum_{s=\kappa+1}^{\infty} p_s \biggr] 
\end{equation}	
So far, $\kappa\in\mathbb{N}$ was fixed arbitrary. Now we choose 
\begin{equation}
	\label{Kdef}
	\kappa = K(E) := \min \Big\{ r \in \mathbb{N}: \sum_{s=r+1}^{\infty} p_s< \frac{E}{2}\Big\}
\end{equation}
and conclude 
\begin{equation}
	\label{l2}
	1 - \N(1-E) \ge \frac{1}{|Q_{K(E)}|} \; \mathbb{P}\biggl[  \frac{1}{|Q_{K(E)}|}
		\sum_{x\in Q_{K(E)}}\omega_x  > - \frac{E}{2}  \biggr].
\end{equation}
Since the \RV{} are assumed to be i.i.d., we obtain
\begin{align}\label{l3}
	\mathbb{P}\biggl[  \frac{1}{|Q_{K(E)}|}\sum_{x\in Q_{K(E)}}\omega_x  > - \frac{E}{2}  \biggr]
	& \geq \mathbb{P}\biggl[ \forall x\in Q_{K(E)}:\; \omega_x > -\frac{E}{2}\biggr] \nn\\ 
	&= \Bigl(\mathbb{P}_{0}[ \omega_{x_0} > - E/2] \Bigr)^{|Q_{K(E)}|} \nn \\
	&= \e^{-g(E)|Q_{K(E)}|}
\end{align}
with $g(E) := - \ln \mathbb{P}_{0}[ \omega_{x_0} > - E/2] $ for every $E>0$. The function $g$ is well-defined because of Assumption \eqref{Ass2} on the tails of the single-site distribution. In fact \eqref{Ass2} implies the estimate
\begin{equation}
	\label{g-bound}
	 0 \le g(E) \le m |\ln(E/2)| - \ln C 
\end{equation}
for all sufficiently small $E>0$ with some $E$-independent constants $C,m >0$.

We have by definition \eqref{Kdef} that $\sum_{s= K(E)}^{\infty} p_{s} \ge E/2$.
In conjunction with the upper bound in \eqref{algdec}, this yields the existence of  $E_{l} >0$ 
such that $\rho^{K(E)} \le C_{l}^{2/d_{s}} E^{-1}$ for all $E \in ]0,E_{l}]$ with the constant $C_{l} := [2C_{2} \rho/(\rho -1)]^{d_{s}/2}$. Since $\rho^{d_{s}/2} =n$, we conclude  
\begin{equation}
	\label{Q-estimate}
	|Q_{K(E)}| \le C_{l} E^{-d_{s}/2}
\end{equation}
for all $E\in ]0,E_{l}]$. Collecting \eqref{l2}, \eqref{l3} and \eqref{Q-estimate}, we finally arrive at 
the desired lower bound
\begin{equation}
	\label{lb}
	1 - \N(1-E) \geq   C_{l}^{-1} E^{d_s/2} \exp\{ - C_{l} E^{- d_s/2} g(E) \} 
\end{equation}
for all $E\in ]0,E_{l}]$.

(c) \quad\emph{Limit $E\searrow 0$.} \quad The lower bound \eqref{lb} can be simplified by observing \eqref{g-bound}, enlarging the constant $C_{l}$ and diminishing $E_{l}$: there exist constants $\wtilde{C}_{l}, \wtilde{E}_{l} > 0$ such that 
\begin{equation}
	1 - \N(1-E) \geq   \wtilde{C}_{l}^{-1} \exp\{ - \wtilde{C}_{l} E^{- d_s/2} |\ln E|\} 
\end{equation}
for all $E\in ]0,\wtilde{E}_{l}]$. This clearly implies
\begin{equation}
 		\liminf_{E\searrow 0} \frac{ \ln \big|\ln \big[1 - \N(1-E) \big]\big|}{\ln E} 
		\ge - \frac{d_s}{2}.
\end{equation}
On the other hand, we deduce from the upper bound \eqref{ub} that
\begin{equation}
 		\limsup_{E\searrow 0} \frac{ \ln \big|\ln \big[1 - \N(1-E) \big]\big|}{\ln E} 
		\le - \frac{d_s}{2},
\end{equation}
and Theorem~\ref{lif-thm} is proven.
\end{proof}

%%%%%%%%%%%%%%%%%%%%%%%%%%%%%%%%%%%%%%%%%%%%%%%%%%%%%%%%%%%%%%%%%%%%%%%%%%%%%%%%%%%
%%%%%%%%%%%%%%%%%%%%%%%%%%%%%%%%%%%%%%%%%%%%%%%%%%%%%%%%%%%%%%%%%%%%%%%%%%%%%%%%%%%

\appendix

\section{Ergodicity}

Here we briefly compile the ergodic structure of the \HAM, which we have not found in the literature. Due to the use of ergodicity, our Lemmas \ref{det-spec} and~\ref{conv} give slightly stronger results, resp.\ have shorter proofs, than corresponding statements in \cite{K1, K2, K3}, who argued without it.

In what follows, we enumerate $\X$ as in \eqref{enumeration} and arrange the elements in increasing order from left to right, see Fig.~\ref{fighs}. Also, we think of $\X$ as being isomorphic to the space  
\begin{equation}
	\label{iso}
 	\X \cong \bigg\{ (\xi_{r})_{r \in\mathbb{N}} : \xi_{r} \in\{ 0,1, \ldots, n_{r} -1\} \text{~for all~} r \in\mathbb{N} \text{~and~} \sum_{r=1}^{\infty}\xi_{r} < \infty \bigg\} ,
\end{equation}
which consists of sequences with only finitely many non-zero elements.
The identification $x=(\xi_{r})_{r \in\mathbb{N}}$, which underlies \eqref{iso}, works as follows: 
$\xi_{1}$ determines the position of $x$ in $Q_{1}(x)$, where $\xi_{1}=0$ corresponds to the 
left-most position in Fig.~\ref{fighs}, $\xi_{1}=1$ to the second position from the left, and so 
on. Similarly, $\xi_{r}$ encodes the position of $Q_{r-1}(x)$ in $Q_{r}(x)$ for every $r \ge 2$, 
where, again, cluster positions are counted from the left, starting with $0$. For example, in Fig.~\ref{fighs}  we have $x_5 = (2,1,0,0,\ldots)$. Every $x\in\X$ 
eventually belongs to $Q_{r_{0}}(x_{0})$, the left-most cluster of rank $r_{0}$, for some sufficiently 
large $r_{0}$. Therefore any sequence $(\xi_{r})_{r \in\mathbb{N}}$ has only finitely many non-zero 
elements. Moreover, we have the representation $x_{k} = (\xi_{r})_{r \in\mathbb{N}}$, if and only 
if $k= \sum_{r=1}^{\infty} \xi_{r} |Q_{r-1}|$.

The countable space $\X$ can be equipped with an Abelian group structure, which we write in an additive way  
\begin{equation}
	\label{group}
	x+ y := \big( (\xi_{r} + \eta_{r} ) \!\!\!\mod n_{r}\big)_{r\in\mathbb{N}}
\end{equation}
for all $x= (\xi_{r})_{r \in\mathbb{N}}$ and $y= (\eta_{r})_{r \in\mathbb{N}}$ in $\X$. 
The identity of the group is given by $x_{0}$ and the inverse of $x$ by $-x := (n_{r} - \xi_{r})_{r \in\mathbb{N}}$.

The discrete Abelian group $\X$ acts on $\Omega$ according to $\X \times \Omega \rightarrow \Omega$, $(x,\omega) \mapsto \tau_{x}(\omega)$, where   
\begin{equation}
  \tau_{x}(\omega) := (\omega_{x+y})_{y \in\X} 
\end{equation}
for every $x\in\X$ and every $\omega := (\omega_{y})_{y\in\X} \in \Omega$. Since $ \mathbb{P} = \bigotimes_{x\in\X} \mathbb{P}_{0}$, every $\tau_{x}$ is measure-preserving, and we have  

\begin{lemma}
 	\label{ergodic-shifts}
 	The group of transformations $\{\tau_{x}\}_{x\in\X}$ is ergodic with respect to $\mathbb{P}$.
\end{lemma}

\begin{proof}
	Let $A \in \bigotimes_{x\in\X} \mathcal{B}_{\mathbb{R}}$ be invariant under $\tau_{x}$ for every $x\in\X$. We show $\PP[A]$ is either $0$ or $1$. Let $\varepsilon >0$ be given. Since every product-measurable set can be approximated arbitrarily well by cylinder sets, there exists  $\kappa \in\mathbb{N}$ and $Z = \bigtimes_{x\in\X} Z_{x}$ with $Z_{x} \in\mathcal{B}_{\mathbb{R}}$ for every $x\in\X$ and $Z_{x} = \mathbb{R}$ for every $x \in \{x_{k} \in\X : k \ge \kappa \}$, and such that 
	\begin{equation}
		\label{approx}
 		\PP[A \sydi Z] \le \varepsilon.
	\end{equation}
	Here $\sydi$ denotes the symmetric difference. Pick $r_{0} \in\mathbb{N}$ such that $x_{\kappa}\in Q_{r_{0}}(x_{0})$ and define $w:= (\delta_{r,r_{0}+1})_{r\in\mathbb{N}} \in\X$. Hence, 
	$\big(w + Q_{r_{0}}(x_{0})\big) \cap Q_{r_{0}}(x_{0})= \emptyset$, which in turn implies the crucial identity
	\begin{equation}
		\label{ergo-start}
 		\PP\big[ \tau_{w}(Z) \cap Z \big] = \PP\big[ \tau_{w}(Z) \big] \;\PP[Z] = \big( \PP[Z]\big)^{2}
	\end{equation}
	because $\tau_{w}(Z) = \bigtimes_{x\in\X} Z_{w+x}$. Using invariance of $A$ under $\tau_{w}$ and 
	\eqref{ergo-start}, we conclude
	\begin{align}
		\label{event-chain}
 		0 \le \PP[A] -\big(\PP[A]\big)^{2} 
		%&= \PP\big[ A \cap \tau_{w}(A) \big] - \big(\PP[A]\big)^{2} 
		%	\nonumber\\
		&= \PP\big[ A \cap \tau_{w}(A) \big] -  \PP\big[ Z \cap \tau_{w}(Z) \big] + 
					\big(\PP[Z]\big)^{2} - \big(\PP[A]\big)^{2} \nonumber\\
			&\le \PP\big[ \big(A \cap \tau_{w}(A)\big) \sydi \big(Z \cap \tau_{w}(Z)\big) \big]
				+ 2 \,\PP[A \sydi Z]  \nonumber\\
			&\le  \PP \big[ (A \sydi Z) \cup \big( \tau_{w}(A) \sydi \tau_{w}(Z)\big)\big]
					+ 2 \,\PP[A \sydi Z]   \nonumber\\
			&\le 4\varepsilon.
	\end{align}
	In the second line of \eqref{event-chain} we estimated $\PP[C] -\PP[D] \le \PP[C \sydi D]$ for events $C$ and $D$, which follows from $C=(C\setminus D) \cup (C \cap D) \subseteq (C \sydi D) \cup D$. The inequality in the third line of \eqref{event-chain} is based upon the inclusion
\begin{equation*}
 (C_{1} \cap D_{1}) \setminus (C_{2} \cap D_{2}) = (C_{1} \cap C_{2}^{c} \cap D_{1} ) \cup (D_{1} \cap D_{2}^{c} \cap C_{1}) \subseteq (C_{1} \setminus C_{2}) \cup (D_{1} \setminus D_{2})
\end{equation*}
and its mirror $1 \longleftrightarrow 2$.	 In order to get to the last  line of \eqref{event-chain} we used  
$\tau_{w}(C) \sydi \tau_{w}(D) = \tau_{w}(C\sydi D)$, the fact that $\tau_{w}$ is measure preserving and \eqref{approx}. Since $\varepsilon > 0$ is arbitrary in \eqref{event-chain}, this completes the proof.
\end{proof}

\noindent
The (F\o{}lner) sequence of growing clusters $\big(Q_{r}(x_{0})\big)_{r\in\mathbb{N}}$ exhausts the (amenable) ergodic group $\X$ and fulfils Shulman's temperedness condition \cite[Def.~1.1]{Lin01}. Therefore we can apply the general pointwise ergodic theorem of Lindenstrauss \cite[Thm.~1.2]{Lin01} and conclude 

\begin{lemma}[Birkhoff ergodic theorem]
	\label{birk}
	For every $\PP$-integrable random variable $h: \Omega \rightarrow \mathbb{C}$ we have 
	\begin{equation}
		\lim_{r\to\infty} \frac{1}{|Q_{r}|} \; \sum_{x\in Q_{r}(x_{0})} h\big(\tau_{x}(\omega)\big)
	 	= \mathbb{E}[h]
	\end{equation}
	for $\PP$-almost every $\omega\in\Omega$.
\end{lemma}

\noindent
The reason for introducing the particular group structure \eqref{group} is its compatibility with the hierarchical structure on $\X$. This underlies 

\begin{lemma} 	
	\label{covariant}
 	The hierarchical Anderson Hamiltonian \eqref{Hamiltonian} transforms covariantly under the group 
	action on $\Omega$,
	\begin{equation}
		\Ha^{\tau_{x}(\omega)} = U_{x}^{*} \Ha^{\omega} U_{x}^{\phantom{*}}
	\end{equation}
	for every $x \in\X$ and every $\omega\in\Omega$. Here we have introduced the unitary representation 
	of the group $\X$ on $\ell^{2}(\X)$, given by $(U_{x}\psi)(y) := \psi(y-x)$ for all $x,y\in\X$ and all $\psi \in \ell^{2}(\X)$.
\end{lemma}

\begin{proof}
 	Clearly we have 
	\begin{align}
 		(U_{x}^{*} V^{\omega} U_{x}^{\phantom{*}}\psi)(y) &= (V^{\omega}U_{x}\psi)(x+y) 
		= \omega_{x+y} (U_{x}\psi)(x+y) \nonumber\\
		&= \big(\tau_{x}(\omega)\big)_{y} \psi(y) = (V^{\tau_{x}(\omega)}\psi)(y),
	\end{align}
	so that it remains to verify the invariance of the hierarchical Laplacian 
	\begin{align}
		\label{delta-cov}
 		(U_{x}^{*} \Delta U_{x}^{\phantom{*}}\psi)(y) &= \sum_{s=1}^{\infty} p_{s} (\E_{s}U_{x}\psi)(x+y) 
		= \sum_{s=1}^{\infty} p_{s} \;\frac{1}{|Q_{s}|} \, \sum_{v \in Q_{s}(x+y)}  (U_{x}\psi)(v) \nonumber\\
		&= \sum_{s=1}^{\infty} p_{s} \;\frac{1}{|Q_{s}|} \, \sum_{v \in x+ Q_{s}(y)}  \psi(v-x)\nonumber\\
		&= \sum_{s=1}^{\infty} p_{s} \;\frac{1}{|Q_{s}|} \, \sum_{w \in Q_{s}(y)}  \psi(w)
			= (\Delta\psi)(y).
	\end{align}
 The equality in the second line of \eqref{delta-cov} rests on the identity 
	\begin{align}
 		Q_{s}(x+y) &= \Big\{ (\zeta_{r})_{r\in\mathbb{N}} : \; \zeta_{r}  \in \{0, \ldots, n_{r}-1\} 
		\quad \text{for~~}	r \le s, \nonumber\\[-1ex]
		 & \hspace*{2.15cm} \zeta_{r} =  (\xi_{r} + \eta_{r} ) \!\!\!\!\mod n_{r} \quad
		 	\text{for~~}	r > s\Big\} 	\nonumber\\
		&= x+ Q_{s}(y)	
	\end{align}
	for all $s\in\mathbb{N}$, $x= (\xi_{r})_{r \in\mathbb{N}}$ and $y= (\eta_{r})_{r \in\mathbb{N}}$ in $\X$. 
\end{proof}

%%%%%%%%%%%%%%%%%%%%%%%%%%%%%%%%%%%%%%%%%%%%%%%%%%%%%%%%%%%%%%%%%%%%%%%%%%%%%%%%%%%
%%%%%%%%%%%%%%%%%%%%%%%%%%%%%%%%%%%%%%%%%%%%%%%%%%%%%%%%%%%%%%%%%%%%%%%%%%%%%%%%%%%


\begin{thebibliography}{MuS2}           
\frenchspacing

\bibitem[AV]{AV} 
	\au{T.}{Antunovi\'c}\et\lau{I.}{Veseli\'c} 
	\ti{Equality of Lifshitz and van Hove exponents on amenable Cayley graphs} 
	\z{J. Math. Pures Appl.}{92}{342--362}{2009} 
                                                                                     
\bibitem[BlS1]{BS} 
	\au{P. M.}{Bleher}\et\lau{Ja. G.}{Sinai}
	\ti{Investigation of the critical point in models of the type of Dyson's hierarchical models} 
	\z{Commun. Math. Phys.}{33}{23--42}{1973}  

\bibitem[BlS2]{BlSi75} 
	\au{P. M.}{Bleher}\et\lau{Ya. G.}{Sinai}
	\ti{Critical indices for Dyson's asymptotically-hierar\-chical models} 
	\z{Commun. Math. Phys.}{45}{247--278}{1975}  

\bibitem[Bo]{Bovier} 
	\lau{A.}{Bovier} 
	\ti{The density of states in the Anderson model at weak disorder: 
			a renormalization group analysis of the hierarchical model}
	\z{J. Stat. Phys.}{59}{745--779}{1990}

%\bibitem [CL]{CL} \textsc{Carmona, R., Lacroix, J.}:\ \textit{Spectral Theory of Random Schr\"odinger Operators}, Birkh\"auser, Boston/Basel/Berlin (1990)

\bibitem[D1]{Dys69} 
	\lau{F. J.}{Dyson} 
	\ti{Existence of a phase-transition in a one-dimensional Ising ferromagnet}
	\z{Commun. Math. Phys.}{12}{91--107}{1969}

\bibitem[D2]{Dys71} 
	\lau{F. J.}{Dyson} 
	\ti{An Ising ferromagnet with discontinuous long-range order}
	\z{Commun. Math. Phys.}{21}{269--283}{1971}

\bibitem[FOR]{FyOsRo09} 
	\au{Y. V.}{Fyodorov}, \au{A.}{Ossipov}\et\lau{A.}{Rodriguez} 
	\ti{The Anderson localization transition and eigenfunction multifractality in an ensemble of ultrametric random matrices}
	\z{J. Stat. Mech.}{}{L12001-1--9}{2009}

\bibitem[Ki]{Kirsch}
	\lau{W.}{Kirsch}
	\ti{An invitation to random Schr\"odinger operators -- with an appendix by F. Klopp} 
	\IN 
	Panor. Synth\`eses, 25, \bti{Random Schr\"odinger operators} 
	\pub[pp. 1--119]{Soc. Math. France}{Paris}{2008}
	
\bibitem[KiMa]{KM}
	\au{W.}{Kirsch}\et\lau{F.}{Martinelli}
	\ti{Large deviations and Lifshitz singularity of the integrated density of 
		states of random Hamiltonians} 
	\z{Commun. Math. Phys.}{89}{27--40}{1983}

\bibitem[KiM\"u]{KiMu06}
  \au{W.}{Kirsch}\et\lau{P.}{M\"uller}
  \ti{Spectral properties of the Laplacian on bond-percolation graphs}
  \z{Math. Z.}{252}{899--916}{2006}

\bibitem[Kr1]{K1} 
	\lau{E.}{Kritchevski}
	\ti{Spectral localization in the hierarchical Anderson model} 
	\z{Proc. Amer. Math. Soc.}{135}{1431--1440}{2007}
	
\bibitem[Kr2]{K2} 
	\lau{E.}{Kritchevski}
	\ti{Hierarchical Anderson Model}
	\IN 
	\bti[CRM Proc. Lecture Notes, 42]{Probability and mathematical physics} 	 
	\pub[pp. 309--322]{Amer. Math. Soc.}{Providence, RI}{2007} 

\bibitem[Kr3]{K3}
	\lau{E.}{Kritchevski}
 	\ti{Poisson statistics of eigenvalues in the hierarchical Anderson model} 
	\z{Ann. Henri Poincar\'e}{9}{685--709}{2008}

\bibitem[L]{Lin01}
  \lau{E.}{Lindenstrauss}
  \ti{Pointwise theorems for amenable groups}
  \z{Invent. Math.}{146}{259--295}{2001}
  
\bibitem[Mo1]{Mol94}
	\lau{S. A.}{Molchanov} 
	\ti{Lectures on random media}
	\IN
	\bti[Lect.\ Notes Math., vol.\ 1581]{Lectures on probability theory} 
	\pub[pp. 242--411]{Springer}{Berlin}{1994}
	
\bibitem[Mo2]{M}
	\lau{S.}{Molchanov} 
	\ti{Hierarchical random matrices and operators. Application to Anderson model}
	\IN
	\ed{A. K.}{Gupta}\et\led{V. L.}{Girko}
	\bti{Multidimensional statistical analysis and theory of random matrices} 
	\pub[pp. 179--194]{VSP}{Utrecht}{1996}
	
\bibitem[MoG]{MoGa11} 
	\au{C.}{Monthus}\et\lau{T.}{Garel} 
	\ti{A critical Dyson hierarchical model for the Anderson localization transition}
	\z{J. Stat. Mech.}{}{P05005-1--27}{2011}

\bibitem[M\"uS1]{MuSt07} 
	\au{P.}{M\"uller}\et\lau{P.}{Stollmann} 
	\ti{Spectral asymptotics of the Laplacian on super-critical bond-percolation graphs}
	\z{J. Funct. Anal.}{252}{233--246}{2007}

\bibitem[M\"uS2]{MuSt11}
	\au{P.}{M\"uller}\et\lau{P.}{Stollmann} 
	\ti{Percolation Hamiltonians}
	\IN
	\ed{D.}{Lenz}, \ed{F.}{Sobieczky}\et\led{W.}{Woess}
	\bti[Progress in Probability, vol. 64]{Random walks, boundaries and spectra} 
	\pub[pp. 235--258]{Springer}{Basel}{2011}
	
% dubios: ist nie erschienen
%
%\bibitem[P]{P}\textsc{Parreira, J.R.}:\ \textit{Phase Transitions in Hierarchical Models:\ 
%The Role of Coupling Coefficients}, Princeton, (1997), www.ma.utexas.edu/mp\_arc/c/96/96-679.ps.gz

%\bibitem [PF]{PF} \textsc{Pastur, L., Figotin, A.}:\ \textit{Spectra of Random and Almost-Periodic Operators}, Springer, Berlin/Heidelberg, (1992)

\bibitem [RS]{RS4} 
	\au{M.}{Reed}\et\lau{B.}{Simon}
	\bti{Methods of modern mathematical physics IV: analysis of operators} 
	\pub{Academic Press}{San Diego}{1978}

\bibitem[S]{Simon}
	\lau{B.}{Simon}
	\ti{Lifschitz tails for the Anderson model}
	\z{J. Stat. Phys.}{38}{65--76}{1985}

\end{thebibliography}
\end{document}